\documentclass{article}

\usepackage{lineno,hyperref}
\modulolinenumbers[5]

\usepackage{amsmath,amssymb,amsfonts,amsthm}
\usepackage{tablefootnote}
\usepackage{xspace}
\usepackage{lettrine,color}
\usepackage{graphicx}
\usepackage{subfigure}
\usepackage{latexsym}
\usepackage{paralist}
\newcommand{\prob}[1]{\mathbf{Pr}[#1]}
\newcommand{\Prob}[1]{\mathbf{Pr}\bigg[#1\bigg]}
\usepackage{algorithm}
\usepackage{algorithmicx}
\usepackage{algpseudocode}
\usepackage[T1]{fontenc}
\usepackage{lmodern}

\newtheorem{theorem}{Theorem}
\newtheorem{corollary}[theorem]{Corollary}
\newtheorem{lemma}[theorem]{Lemma}

\newtheorem{definition}[theorem]{Definition}

\newcommand{\net}{\ensuremath{\mathfrak{N}}\xspace}	
\newcommand{\Exp}[1]{\mathbf{E}\left[#1\right]}

\newcommand{\cj}{\ensuremath{c}\xspace}

\title{Communicating with Beeps\footnote{Research partially supported by the Centre for Discrete Mathematics and its Applications (DIMAP).}}

\begin{document}
	

		
\author{\textbf{Artur Czumaj} \hspace{4mm} \textbf{Peter Davies} \\[0.10in]
	Department of Computer Science \\
	Centre for Discrete Mathematics and its Applications 
	\\
	University of Warwick}

\title{\textbf{Communicating with Beeps}
	\thanks{Research partially supported by the Centre for Discrete Mathematics and its Applications (DIMAP).}
	\thanks{Contact information: \{A.Czumaj, P.Davies.4\}@warwick.ac.uk. Phone: +44 24 7657 3796.}
}

\maketitle
		
\begin{abstract}
The \emph{beep model} is a very weak communications model in which devices in a network can communicate only via beeps and silence. As a result of its weak assumptions, it has broad applicability to many different implementations of communications networks. This comes at the cost of a restrictive environment for algorithm design.

Despite being only recently introduced, the beep model has received considerable attention, in part due to its relationship with other communication models such as that of ad-hoc radio networks. However, there has been no definitive published result for several fundamental tasks in the model. We aim to rectify this with our paper.

We present algorithms and lower bounds for a variety of fundamental global communications tasks in the model.
\end{abstract}

\section{Introduction}

The \emph{beep model}, introduced recently by Cornejo and Kuhn \cite{-CK10}, is a very weak network communications model in which information can be passed only in the form of a beep or a lack thereof. The model is related to the ad-hoc radio network model, and has been used as a surrogate model in results concerning radio networks with collision detection. As well as attracting study from this angle, the beep model is interesting in its own right because of its generality, simplicity, and wide range of areas where it could be applied.

Despite being only recently introduced, the beep model has received considerable attention, in part due to its relationship with other communication models such as that of ad-hoc radio networks. However, there has been no definitive published result for several fundamental communication tasks in the model. We aim to rectify this with our paper.

\subsection{Model}

The network is modeled as an undirected connected graph $G=(V,E)$, where vertices in the graph represent devices in the network, and edges represent direct reachability. Time is divided into discrete steps, with a synchronized global clock. In each time-step every node decides whether to \emph{beep} or to \emph{listen}. Nodes which choose to listen in a particular time-step hear a beep if at least one of their neighbors chose to beep, and they cannot distinguish between one neighbor beeping or many. We will assume that nodes have unique labels (IDs), which is essential (at least when considering deterministic algorithms) in order to break symmetry.

We will use the following parameters in analysis of our algorithms:
\begin{itemize}
\item $n$ will denote network size, i.e., $|V|$.
\item $D$ will denote network diameter, the largest distance between any pair of nodes.
\item $L$ will be the range of node labels, i.e., labels will be strings of no more than $\log L$ bits.
\item $M$ will be the range of messages, i.e., messages will be strings of no more than $\log M$ bits.
\item $k$ will be the number of source nodes when considering the multi-broadcast task.
\end{itemize}

We do not, however, assume that nodes have any prior knowledge of these parameters, nor any other knowledge about the network.

\subsection{Related Work}

There has been a large amount of research focusing on fundamental communication problems in distributed computing, see e.g., \cite{-Pel07} and the references therein. The beep model was introduced by Cornejo and Kuhn \cite{-CK10}, who used it to design an algorithm for interval coloring. This task is a variant of vertex coloring used in resource allocation problems, and is, in a sense, tailored to the model. In another recent work, Afek et al.\ \cite{-AABCHK11} presented an algorithm for finding a maximal independent set in the beep model, and an algorithm for the related problem of minimum connected dominating set is given in \cite{-YJYLC15}.

The beep model is strictly weaker than the model of radio networks with collision detection (see, e.g., \cite{-Pel07}), and so algorithmic results in the former also apply in the latter. This relationship was exploited by Ghaffari and Haeupler \cite{-GH13} to give almost optimal $O((D+\log n \log\log n) \cdot \min \{\log\log n, \log \frac nD\})$-time randomized algorithm for leader election in radio networks with collision detection. Ghaffari and Haeupler \cite{-GH13} also introduce the method of ``beep waves'' to transmit bit strings, a method which is also employed here for the purpose of broadcast. Czumaj and Davies \cite{-CD16} give a different randomized leader election algorithm which achieves optimal $O(D+\log n)$ expected running time at the cost of slower worst-case time. Ghaffari et al.\ \cite{-GHK13} give a randomized broadcast algorithm in radio networks with collision detection which employs beeping techniques, but, unlike the algorithm of \cite{-GH13}, does not entirely translate over to the beep model.

A deterministic leader election algorithm in the beep model was given by F\"{o}rster et al.\ \cite{-FSW14}, taking $O(D \log L)$ time, and a very recent result by Dufoulon, Burman and Beauquier \cite{-DBB18} improves this to an optimal $O(D+\log L)$ time.   In another related work, Gilbert and Newport \cite{-GN15} studied the quantity of computational resources needed to solve specific problems in the beep model.

Concurrently with this work, Hounkanli and Pelc \cite{-HP15} give a $O(D+\log M)$ time broadcasting algorithm and an $O(n^2\log M + nD \log L)$-time gossiping algorithm in a slightly different model where nodes know network parameters $n, L, M$ but wake-up at arbitrary different time-steps, rather than simultaneously. To our knowledge there have been no earlier published results for broadcast, gossiping, and multi-broadcast in the model we study. Further recent works explore other tasks in various related beeping models (\cite{-BKKPW16,-CMT17,-HMP16,-HP16}).

\subsection{Our Results}

Our aim here is to provide the first comprehensive study of global communication algorithms in the beep model. Recall that $n$ denotes number of nodes, $D$ diameter, $M$ message range, $L$ label range, and $k$ number of sources for multi-broadcast. We present the following results:

\begin{itemize}
	\item An optimal $O(D+\log M)$-time algorithm for broadcasting a $\log M$ bit message, developing and formalizing the ``beep waves'' method of \cite{-GH13}.
	\item A corresponding  $\Omega(D+\log M)$ lower bound.
	\item An optimal $O(D+\frac{D\log M}{\log D})$-time algorithm for broadcasting a $\log M$ bit message in directed networks.
	\item A corresponding  $\Omega(D+\frac{D\log M}{\log D})$ lower bound.
	\item An $O(k \log \frac{LM}{k}+D\log L)$-time explicit algorithm, and an optimal $O(k \log \frac{LM}{k}+D)$-time non-explicit algorithm for multi-broadcast with provenance (where every node must learn all (source ID, source message) pairs).
	\item A corresponding $\Omega(k \log \frac{LM}{k}+D)$ lower bound.
	\item An explicit algorithm for multi-broadcast without provenance (where every node must learn all unique source messages) taking $O(k \log \frac{M}{k}+D\log L)$ time when $M>k$ and $O(M+D\log L)$ time when $M\le k$.
	\item A non-explicit algorithm for multi-broadcast without provenance taking $O(k \log \frac{M}{k}+D+\log L)$ time when $M>k$ and $O(M+D+\log L)$ time when $M\le k$.
	\item A corresponding lower bound of $\Omega(k \log \frac{M}{k}+D)$ when $M>k$ and $\Omega(M+D)$ when $M\le k$.
\end{itemize}

These multi-broadcasting algorithms imply $O(n \log \frac{LM}{n})$ and $(n \log \frac{M}{n}+\log L)$-time gossiping algorithms with and without provenance respectively.

\section{Broadcasting}
The first, and most basic, task we will consider in the beep model is that of broadcasting, where a source node begins with a message of which to inform all other nodes. Since messages must, in effect, be transmitted ``bit by bit'' in a pattern of beeps and silence, algorithmic running time is affected by the length of the message we must transmit (this is not generally the case in standard radio networks, where we assume the message can be passed in a single transmission). So, we introduce a new parameter $M$ to specify message range, and assume that all messages to be broadcast are integers in $[M]$.

\subsection{Broadcasting in Undirected Networks}

Broadcasting in undirected networks will be performed using a method known as `beep waves'. Beep waves were first introduced by Ghaffari and Haeupler \cite{-GH13} as a means of transmitting information in the beep model. Variations of the technique are useful for different circumstances, and here we give a simple formalization tailored to the task of broadcasting from a single source.

The idea is the following: every three time-steps, starting at zero, the source transmits a bit of its message, that is it beeps to represent a \textbf{1} or remains silent to represent a \textbf{0}. All other nodes aim to relay any beep coming from a neighbor one hop closer to the source, in the next time-step after they hear it. Of course, nodes do not know the provenance of beeps they hear, but we can ensure that nodes will not hear any beeps from their own layer since they will themselves be beeping rather than listening. We can also stipulate that nodes become `deaf' and ignore any beeps they hear in time-steps immediately after they transmitted themselves, and this rules out beeps from the next layer. Then, nodes will only relay beeps from the previous layer, so the waves of beeps will emanate out from the source, one distance hop per time-step, and inform all nodes of the source message.

\begin{algorithm}[h]
	\caption{\textsc{Beep-Wave}$(s,m)$ at source $s$}
	\label{alg:BW1}
	\begin{algorithmic}
		\State s beeps at time-step $0$
		\For {$t = 1$ to $|m|$}
		\State \textbf{if} bit $m_t$ is \textbf{1} \textbf{then} $s$ beeps in time-step $3t$
		\EndFor
	\end{algorithmic}
\end{algorithm}
\addtocounter{algorithm}{-1}
\begin{algorithm}[h]
	\caption{\textsc{Beep-Wave}$(s,m)$ at non-source $u$}
	\begin{algorithmic}
		\State $j \gets$ first time-step $u$ hears a beep
		\While {end of message not heard}
		\If{$u$ hears a beep in time-step $t \equiv j \bmod 3$}
		\State{$u$ beeps in time-step $t+1$}
		\State bit $m(u)_{\frac{ t-j}{3 }} \gets 1$
		\EndIf
		\EndWhile
		\State output $m(u)$
	\end{algorithmic}
\end{algorithm}

\begin{theorem}
	\label{thm:BW1}
	\textsc{Beep-Wave}$(s,m)$ correctly performs broadcast in time $O(D+|m|) = O(D+\log M)$.
\end{theorem}

\begin{proof}
	Partition all nodes into layers depending on their distance from the source $s$, i.e., layer $L_i = \{v\in V : dist(v,s) = i\}$. We show that a node in layer $L_i$ beeps in time-step $t$ iff $3|t-i$ and either $m_\frac{t-i}{3} = 1$ or $t=i$, by induction on $t$.
	
	For $t=0$, the claim is trivially true, since the source $s\in L_0$ beeps, and all $u$ in later layers do not.
	
	For $t=t'>0$, the claim is again clearly true for the source $s$. Consider a non-source node $u\in L_i$, with $i\leq t'$. Such a node hears its first beep, from a neighbor in layer $L_{i-1}$, at time-step $i-1$ by the inductive assumption, and so sets $j=i-1$. Node $u$ can only beep in time-step $t'$ if $t' \equiv i \bmod3$, and in this case it beeps only upon hearing a beep in time-step $t' - 1$ (which, by the inductive assumption, can only come from a node in layer $L_{i-1}$). So, again by the inductive assumption, $m_\frac{(t'-1)-(i-1)}{3} = m_\frac{t'-i}{3} = 1$, i.e. $u$ beeps if and only if the correct conditions are satisfied.
	
	When $u$ beeps in time-step $t$, $m(u)_{\frac{ t-1-j}{3 }} = m(u)_{\frac{ t-i}{3 }}$ is set to $1$. So, $m_\frac{t-i}{3} = 1\iff m(u)_{\frac{ t-i}{3 }}  = 1$, i.e. $u$'s output message is correct.

	By induction the claim is true for all $t$, and so $m_\frac{t-i}{3} = 1 \iff m(u)_{\frac{ t-i}{3 }}  = 1$. Furthermore, after $D+3\log M$ time-steps, all nodes cease transmission.
	
\end{proof}

This is, to our knowledge, the first formalization of beep waves for the task of broadcasting, and the first efficient beeping algorithm for the task.

\subsection{Broadcasting in Directed Networks}

Allowing the underlying graph of the network to be directed greatly restricts what can be done efficiently in the beep model. Beep waves as described above, which are the basis of almost all efficient beeping algorithms, do not work on directed graphs since nodes cannot distinguish between new waves from the source and `backtracking' from further out layers. In particular, a beep-wave moving through the network can flood all previously reached layers with beeps every time-step, preventing any further communication until it is completed.

Despite these difficulties, we present an algorithm which broadcasts a message in $[M]$ within an optimal $O(\frac{D\log M}{\log D})$ time-steps. We assume throughout that $M\ge D$ (and this is necessary for the running time, since $\Omega(D)$ is a lower bound for broadcasting).

We first give an algorithm which assumes knowledge of $D$ (Algorithm \ref{alg:DB1}), and then describe how it can be extended to remove this assumption. To allow this subsequent extension, we will design Algorithm \ref{alg:DB1} to broadcast from a set $S$ of sources rather than a single source.

\begin{algorithm}[h]
	\caption{\textsc{DirectedBroadcast}$(m,D)$ at source $s\in S$}
	\label{alg:DB1}
	\begin{algorithmic}
		\State{beep in time-step $0$}
		\For{j from $1$ to $\frac{\log M}{\log D}$}
		\State{interpret bits $j\log D$ to $(j+1)\log D - 1$ of $m$ as an integer $x_j\in[0,D-1]$}
		\State{beep $x_j+D+1$ time-steps after previous beep sent}
		\EndFor
		\State{beep $2D+1$ time-steps after previous beep sent}
	\end{algorithmic}
\end{algorithm}
\addtocounter{algorithm}{-1}
\begin{algorithm}[h]
	\caption{\textsc{DirectedBroadcast}$(m,D)$ at non-source $u$}
	\begin{algorithmic}
		\State when $u$ first hears a beep in time-step $i$, it beeps in time-step $i+1$
		\Loop
		\If{$u$ hears a beep in time-step $t$}
		\State{$u$ beeps in time-step $t+1$}
		\State $x \gets$ number of time-steps since last beep heard
		\State \textbf{if} $x \le 2D$ \textbf{then} append $x-D-1$ as a bit-string to $m(u)$
		\State \textbf{else} output $m(u)$
		\State \textbf{end if}
		\State{$u$ becomes \emph{deaf} until time-step $t+D+1$}
		\EndIf
		\EndLoop
		\State output $m(u)$
	\end{algorithmic}
\end{algorithm}

The idea of this algorithm is still similar to beep-waves, in that beeps propagate out from the source set one distance layer per time-step. However, these waves could interfere with any layer they have already passed at any later time-step, so the waves cannot be pipelined as before, and we must instead wait $D$ time-steps for the wave to complete before anything more can be done. This is the purpose of nodes becoming \emph{deaf} for $D$ time-steps after relaying a beep; by this we mean that even if nodes hear beeps, they act as if they did not.

Since we cannot pipeline the waves, we instead use their timing to convey additional information; the source set must wait at least $D+1$ time-steps between waves, but if we allow it to choose any delay between $D+1$ and $2D$ then it can use these $D$ options to convey $\log D$ bits of the message. In this way we improve run-time by a factor of $\log D$ over the naive approach (of using beep-waves with $D$ time-steps delay).

\begin{lemma}
	Algorithm \ref{alg:DB1} performs broadcast from a set of sources $S$ in $O(\frac{D\log M}{\log D})$ time when $D$ is known.
\end{lemma}

\begin{proof}
	Similarly to our analysis of Algorithm \ref{alg:BW1}, we divide nodes into layers based on their distance from the source set, i.e. layer $L_i := \{v\in V : \min_{s\in S} dist(v,s) = i\}$. As before, nodes hear their first beep in time-step $i-1$, and first beep themselves in time-step $i$. 
	
	Let $m'$ be the bit-string transmitted by the sources, i.e. with $m'_0 = 1$ and $1$s placed at each interval $x_j$, where $x_j$ is the integer value of the $j^{th}$ block of $\log D$ message bits, as described. We prove that a node $v\in L_i$ beeps in time-step $t$ iff $m'_{t-i}=1$, by induction on $t$.
	
	The base case $t=0$ is obvious, since sources beep and non-sources do not, as required. Indeed, source nodes clearly have the correct behavior in all time-steps. For the inductive step $t=t'$, we examine a non-source node $v \in L_i$ and divide into two cases:
	
	\textbf{Case 1:} $m'_{t'-i}=1$, i.e. $v$ should beep. In this case, by the inductive assumption, in time-step $t'-1$ all nodes in $L_{i-1}$ beep, including a neighbor of $v$, so we need only show that $v$ is not \emph{deaf} at this time. This is the case, since, again by the inductive assumption, $v$ became deaf the last time it beeped (at time-step $\tilde t := t'-1-x_j$ for appropriate $j$), and no nodes in layers $L_{\geq i-1}$ have beeped between time-step $\tilde t + D$ and $t'-2$.
	
	\textbf{Case 2:} $m'_{t'-i}=0$, i.e. $v$ should not beep. If $v$ has beeped since time-step $t'-(D+1)$ steps then it will be \emph{deaf} and will not beep. Otherwise, the last time-step in which $v$ beeped (again denoted $\tilde t := t'-1-x_j$ for appropriate $j$) satisfies $\tilde t < t' - (D+1)$, in which case by the inductive assumption no node in layers $L_{\geq i-1}$ beep in time-step $t'-1$, so $v$ is silent in time-step $t'$ as required.
	
	Having proven that the beeping behavior of each node is as expected, it is easy to see that nodes can correctly reconstruct the intervals $x_j$ and therefore the message $m$ from their beeping pattern. Furthermore, all nodes cease beeping after at most $D+ 2D\frac{\log M}{\log D} = O(D\frac{\log M}{\log D})$ time-steps.
\end{proof}

This algorithm requires knowledge of $D$. However, it is easy to see that this assumption can be removed by using a doubling technique. Since nodes know their distance from the source after receiving their first beep, we can have them partition themselves into groups based on an exponentially increasing distance range, i.e., group $i$ consists of nodes of distance between $2^i$ and $2^{i+1}$ from the source. Then, we simply perform the algorithm in sequence for each group, with the closest distance layer in the group as the source set and the width of the group as the value for $D$. 

\begin{theorem}
	There is an algorithm which performs broadcasting in a directed network in the beep model in $O(\frac{D\log M}{\log D})$ time, without knowledge of network parameters.
\end{theorem}

\begin{proof}
	Consider an application of Algorithm \ref{alg:DB1} to a group $i$ (of width $2^i$) as described above. Every beep propagated through the group informs the nodes of $\log 2^i = i$ bits of the message, so after $\frac {\log M}{i}$ rounds broadcast is completed within the group. Each round takes at most $2^{i+1}$ time-steps, so the total time to broadcast within the group is $O(\frac{2^i \log M}{i})$. Therefore broadcasting is completed in the whole network within $O\left(\sum\limits_{i=1}^{\log D} \frac{2^i \log M}{i}\right)$ time. This can be bounded as follows:
	
	\begin{align*}
	\sum_{i=1}^{\log D} \frac{2^i \log M}{i}
	&\le
	\log M \left(\sum_{i=1}^{\frac{\log D}{2}} \frac{2^i}{i} +
	\sum_{i=\frac{\log D}{2}}^{\log D} \frac{2^i}{i}\right)
	\\
	&\le
	\log M \left(\sum_{i=1}^{\frac{\log D}{2}} 2^i +
	2\sum_{i=\frac{\log D}{2}}^{\log D} \frac{2^i}{\log D}\right)
	\\
	&\le
	\log M \left(2\sqrt D + \frac{4D}{\log D}\right)
	=
	O\left(\frac{D\log M}{\log D}\right)
	\enspace.
	\qedhere
	\end{align*}
\end{proof}

\section{Multi-Broadcast}
In this section we present our algorithms for the more complex task of multi-broadcast, in undirected networks.

\subsection{Auxiliary Tasks}
Our multi-broadcast algorithms will have a modular structure, i.e. we will use several sub-procedures to solve simpler tasks. We detail these tasks, and the algorithms we will use to solve them:

\subsubsection{Broadcasting}
The multi-broadcasting algorithms we present will, as one might expect, use single-source broadcasting as a sub-routine, and for this we can make use of \textsc{Beep-Wave} (Algorithm \ref{alg:BW1}). Since we must perform several broadcasts with several different messages, however, we must take care to ensure that these are distinguishable. This can be done by encoding the message so that it is obvious when the beginning and end are, for example by duplicating every bit of the message and then placing \textbf{10} at the beginning and end. Note that this coding method does not increase the asymptotic length, in bits, of the message, and that we can decode to find the original message(s), even if there are several, separated by any number of $\textbf{0}$s. We will henceforth assume that all source messages will be encoded in this way.

Algorithm \ref{alg:BW1} only functions correctly when called with a \emph{single} source node, and so we must somehow have the network agree on which node this source should be. To achieve this agreement, we will use an existing algorithm for \emph{leader election}.

\subsubsection{Leader Election}
Leader election enables all nodes to agree on the ID of one particular node to designate leader. In our applications, we will always choose the node with the highest ID in the entire network. More generally, though, leader election can be used on any subset of nodes, whenever each holds some integer value, to find the participating node with the highest (or lowest) such value. The values need not even be unique, since if multiple nodes hold the target value, we can pick out one by performing leader election again on their IDs. 

We wish to be able to perform leader election in $O(D\log L)$ time. If we assume parameter knowledge, there is a straightforward way to do this: we can perform a binary search for the highest ID, iterating through the bits of the IDs and having all nodes who are still ``in the running'' for leader, and who have a 1 in the current position, broadcast. While we cannot use our previous broadcast procedure with multiple sources, since these nodes need only transmit a single bit we can still use beep-waves to ensure that the network hears \emph{something}. This is sufficient for all nodes to determine whether any have a 1 in the current position. A similar method to this was used to perform leader election in radio networks in \cite{-CGR00}.

Without parameter knowledge, however, the task is much more difficult, since without estimates of how long broadcasting, for example, will take, we cannot globally co-ordinate node behavior. Fortunately, one of the few existing results in the beep model is an algorithm by F\"orster, Seidel, and Wattenhofer \cite{-FSW14} that achieves this:

\begin{theorem}
	There is an algorithm \textsc{ElectLeader} which performs leader election in time $O(D \log L)$ without prior knowledge of $D$ or $L$. \qed
\end{theorem}

Furthermore, upon completion, all nodes have knowledge of the highest ID, and can therefore use this as $L$ in future operations.

To perform further tasks after leader election, we require that nodes should know that leader election is complete and that the next stage should begin. While the leader election algorithm \cite{-FSW14} does not immediately allow all nodes to agree on a time-step when this is the case, it does provide the property that the leader is aware of a time-step $t=O(D\log L)$ for which all nodes at distance $i$ from the leader have finished leader election by time-step $t+i$. That is, a procedure commencing with a beep wave from the leader at time-step $t$ will execute successfully. The procedure for diameter estimation we now describe has precisely this property.

\subsubsection{Diameter Estimation}

Our model assumes that nodes do not have access to any of the network parameters. In algorithms for complex tasks, we generally wish to start with a leader election phase, and this provides all nodes with knowledge of $L$. However, if we also wish to know the value of $D$, we must perform an extra task for this purpose.

Our diameter estimation procedure (Algorithm \ref{alg:ED1}) works as follows: we take as input a leader node to co-ordinate the process. An initial beep from the leader propagates through the network. Having received this beep, nodes beep to acknowledge their existence back to the leader; a modularity restriction on when nodes can transmit ensures that these beeps only travel backwards through the layers. While the initial beep from the leader is still reaching further nodes, acknowledgment beeps will continue to return through the network every three time-steps. Once all nodes have been reached, this pattern will cease, and the leader will know the distance of the furthest node, and hence a $2$-approximation of diameter. All of the other nodes have also ceased transmission, and so an application of \textsc{Beep-Wave} can safely be used to broadcast the diameter estimate.

We split the algorithm into two parts, one performed by the leader, and one performed by all non-leader nodes, since their behavior is quite different.

\begin{algorithm}
	\caption{\textsc{EstimateDiameter}$(v)$ at leader $v$}
	\label{alg:ED1}
	\begin{algorithmic}
		\State $v$ beeps in time-step $2$
		\State{let $t$ be the first time-step (greater than $3$) in which $v$ has not received a beep \\
			\qquad for $3$ previous time-steps}
		\State let $\tilde{D} = \frac{2t-8}{3}$
		\State perform \textsc{Beep-Wave}$(v,\tilde{D})$
		\State output $\tilde{D}$
	\end{algorithmic}
\end{algorithm}
\addtocounter{algorithm}{-1}
\begin{algorithm}
	\caption{\textsc{EstimateDiameter}$(v)$ at non-leader $u$}
	\begin{algorithmic}
		\State let $j$ be the first time-step in which $u$ receives a beep
		\State $u$ beeps in time-step $j+2$
		\While{$u$ has heard a beep in the last $3$ time-steps}
		\State{any beep $u$ hears in a time-step equivalent to $j+1 \bmod 3$, \\
			\qquad\qquad it relays in the next time-step}
		\EndWhile
		\State $\tilde{D} \gets \textsc{Beep-Wave}(v,\tilde{D})$
		\State output $\tilde{D}$
	\end{algorithmic}
\end{algorithm}

\begin{lemma}
	\label{lem:ED}
	\textsc{EstimateDiameter} correctly broadcasts an estimate $\tilde{D}$ satisfying $D \le \tilde{D} \le 2D$, and terminates within $O(D)$ time-steps.
\end{lemma}

\begin{proof}
	The first part of the algorithm, in which the leader $v$ beeps in time-step $2$ and other nodes relay beeps after two steps, is effectively a beep wave propagating outwards from the leader one hop per two time-steps. It is easy to see that a node at distance $i$ from the leader receives its first beep in time-step $2i$, and so sets $j=2i$. Furthermore, any node in of distance $i+1$ receives its first beep in time-step $2i+2$, and subsequently beeps itself in time-step $2i+4 \equiv 2i+1 \bmod 3$. This meets the modularity requirement for a node at distance $i$ to relay the beep in the next time-step. Indeed, in general, the modularity requirement ensures that nodes always relay beeps received from the nodes 1 hop further from the leader, and never relay beeps from nodes 1 hop nearer, or the same distance. So, the effect is that a beep wave is sent back to the leader, every three time-steps, by nodes as they are reached by the initial wave.
	
	When the leader no longer receives these beep waves (i.e. as soon as 3 consecutive time-steps occur with no beep heard), it can conclude that all nodes have been reached by the initial beep-wave and have sent a beep-wave back. 
	
	Let $D'$ be the distance from the leader to some the furthest node $u$. Then, $D\le 2D' \le 2D$. The leader emits a beep in time-step $2$ which travels to this furthest node in time-step $2D'$. Node $u$ then beeps in time-step $2D'+2$, and this beep is relayed back to the leader in time-step $3D'+1$. After another $3$ time-steps, the leader knows that it has received the final acknowledgment beep, and sets $t=3D'+4$, making its diameter estimate $\tilde{D} = 2D'$. Hence, as required, $D\le \tilde{D}\le 2D$.
	
	To analyze running time, notice that the leader $v$ reaches its estimate $\tilde D$ in $3D'+4 = O(D)$ time-steps, and the final beep-wave of this value takes also takes $O(D + \log D)$ = $O(D)$ time.
\end{proof}

Since we are only interested in asymptotic behavior, we will assume, for ease of notation, that having performed \textsc{EstimateDiameter} as part of a more complex algorithm we can then make use of the exact value of $D$. Furthermore, once leader election and diameter estimation are performed, all nodes have common linear estimates of $D$ and $L$ and so can agree on a time-step in which both tasks are complete and further procedures can commence.

\subsubsection{Message Collection}

We next introduce a sub-procedure (Algorithm \ref{alg:CM1}) which will allow the leader to collect messages $m(S)$ from a set of sources $S$, receiving an \textbf{OR}-superimposition of all the messages. This works similarly to the usual beep-waves procedure, except that nodes use their distance from the leader (inferred by the time taken to receive the initial \textsc{Beep-Wave}$(v,\textbf{1})$) to ensure that the waves only travel towards the source, and all messages arrive at the same time. We must have an input parameter $p$ giving an upper bound on the length of messages, so that nodes know when the procedure is finished, and we assume that we have already performed \textsc{EstimateDiameter} and so can make use of $D$. We denote by $dist(u)$ the distance from $u$ to the leader node $v$, which can be determined during an application \textsc{Beep-Wave}$(v,\textbf{1})$.

\begin{algorithm}
	\caption{\textsc{CollectMessages}$(v,S,m(S),p)$ at node $u$}
	\label{alg:CM1}
	\begin{algorithmic}
		\State perform \textsc{Beep-Wave}$(v,\textbf{1})$
		\For {$j= 0$ to $p$}
		\If{$m(u)_{j}=1$ or $u$ hears a beep in time-step $D-dist(u)+3j-1$}
		\State{$u$ beeps in time-step $D-dist(u)+3j$}
		\State \textbf{if} $u$ = $v$ \textbf{then} bit $m(u)_{(j-D)/3} \gets 1$
		\EndIf
		\EndFor
		\State output $m(v)$
	\end{algorithmic}
\end{algorithm}

\begin{lemma}
	\textsc{CollectMessages}$(v,S,m(S),p)$ correctly informs $v$ of the \textbf{OR}-superimposition of $m(S)$ within $O(D+p)$ time-steps
\end{lemma}

\begin{proof}
	It is clear that (excluding the initial beep wave) a node $u$ at distance $dist(u)$ from the leader $v$ only ever beeps in time steps equivalent to $D-dist(u) \bmod 3$. Furthermore, nodes only relay beeps they hear in time-steps equivalent to $D-dist(u)-1 = D-(dist(u)+1) \bmod 3$, i.e. they only relay beeps from nodes one hop further than them from the leader. So, if any source node $s$ has $m(s)_j=1$ for some $j$, it beeps in time-step $D-dist(u)+3j$, and this is relayed back to the leader one distance-hop per time-step. The leader $v$ beeps in time-step $D-dist(u)+3j+dist(u)=D+3j$, and hence correctly sets $m(v)_j=1$.
	
	The running time for the initial beep wave is $D$ steps, and for the loop is $3p+D$. So, total running time is $O(D+p)$.
\end{proof}

\subsubsection{Message Length Determination}

One issue with using \textsc{CollectMessages} is the necessity of prior knowledge of a common upper bound on message size. We give a simple method of obtaining this bound (Algorithm \ref{alg:GML}).

We perform \textsc{CollectMessages} using strings which are as long as the messages we actually want to collect, but consist of entirely \textbf{1}s. The superimposition of these strings is a \textbf{1}-string of equal length to the longest message. Since the leader will be able to tell that this string has ended when it hears the substring \textbf{10}, the procedure can be terminated even without an upper bound for the \textsc{CollectMessages} call.

\begin{algorithm}
	\caption{\textsc{GetMessageLength}$(v,S,m(S))$}
	\label{alg:GML}
	\begin{algorithmic}
		\State{perform $p \gets$ \textsc{CollectMessages}$(v,S,\textbf{1}^{m(S)},\infty)$, terminating \\
			\qquad when $v$ hears the substring~\textbf{10}}
		\State perform \textsc{Beep-Wave}$(v,|p|)$
		\State output $|p|$
	\end{algorithmic}
\end{algorithm}

\begin{lemma}
	\textsc{GetMessageLength}$(v,S,m(S))$ correctly informs all nodes of $q = \max_{s\in S}|m(s)|$ within $O(D+q)$ time-steps
\end{lemma}

\begin{proof}
	\textsc{CollectMessages} will terminate after $D+3q$ steps, since $v$ will hear the final \textbf{1} and then a \textbf{0}. All other nodes will be inactive and so \textsc{Beep-Wave}$(v,|p|)$ will successfully inform the network of $q$ (nodes will be aware that the \textsc{CollectMessages} phase is over and so perform \textsc{Beep-Wave} correctly, since they either heard a string of contiguous \textbf{1}s and then a \textbf{0} during \textsc{CollectMessages}, or silence for more than $D$ time-steps).
	
	Running time is $O(D+q)$ for \textsc{CollectMessages} and $O(D+\log q)$ for \textsc{Beep-Wave}, giving $O(D+q)$ total.
\end{proof}

\subsection{Explicit Multi-Broadcast Algorithms}

We are now ready to combine these sub-procedure to perform multi-broadcast. Recall that we consider two variants of the problem: \emph{multi-broadcast with provenance}, where the network must become aware of all (source ID, source message) pairs, and \emph{multi-broadcast without provenance}, where the IDs need not be known. 

\subsubsection{Multi-Broadcast With Provenance}

We first present an algorithm for multi-broadcast with provenance, where all nodes must be made aware of not only the source messages, but also the IDs of the sources they originated from.

The idea of the algorithm is essentially to conduct $k$ simultaneous binary searches to allow a leader to ascertain the IDs of all sources. The process consists of $\log L$ rounds, one for each bit of the IDs. Each node will maintain a list of known prefixes of source IDs, and we aim to preserve the invariant that, after round $i$, all nodes know the first $i$ bits of every source ID. We denote the number of distinct known prefixes at the start of round $i$ by $k_i$.

At the start of round $i$, sources know $k_{i}$ distinct $i-1$-bit ID prefixes (note $k_{i}$ may be less than $k$, since some IDs may share prefixes), and they will each construct a $2k_{i}$-bit string in which each bit corresponds to a particular $i$-bit prefix. Specifically, if we denote the known prefixes in lexicographical order by $(p_1,p_2,\dots , p_{k_{i}})$, then bit $2j$ in the new string will represent the prefix $p_j\textbf{0}$, and bit $2j+1$ will represent $p_j\textbf{1}$. Each source constructs its string by placing a \textbf{1} in the position corresponding to its own ID's $i$-bit prefix, and \textbf{0} in all others. We will denote the string constructed in this manner by source $s$ in round $i$ by $Z_{s,i}$.

Performing \textsc{CollectMessages} with these strings ensures that the leader receives the \textbf{OR}-superimposition, which informs it of all $i$-bit prefixes of source IDs (since it is aware of which prefix each position corresponds to). It then broadcasts this back out to the network via the standard beep wave procedure, and thus the invariant is fulfilled round $i$. After $\log L$ rounds, the IDs of all sources are known in entirety by all nodes. We then perform one final \textsc{CollectMessages} procedure, this time to collate all of the messages the sources wish to broadcast to the network. We construct a $k\log M$-bit string in which the $j^{th}$ block of $\log M$ bits corresponds to the message of the $j^{th}$ source (in lexicographical order of ID). Each source individually fills in its own message in the appropriate block, leaving all other bits as \textbf{0}. We denote the string constructed in this manner by source $s$ as $\tilde{m}_s$. Performing \textsc{CollectMessages} on these strings ensures that the full string of messages arrives at the leader, who then broadcasts it back out to the network.

\begin{algorithm}[h]
	\caption{\textsc{Multi-Broadcast With Provenance}$(S,m(S))$}
	\label{alg:MB1}
	\begin{algorithmic}
		\State $v \gets \textsc{ElectLeader}$
		\State $D \gets \textsc{EstimateDiameter}(v)$
		\State $\log M \gets \textsc{GetMessageLength}(v,S,m(S))$
		\For {$i = 1$ to $\log L$}
		\State $Z_i \gets \textsc{CollectMessages}(v,S,Z_{S,i}, 2k_i)$
		\State perform \textsc{Beep-Wave}$(v,Z_i)$
		\EndFor
		\State $\tilde{m} \gets \textsc{CollectMessages}(v,S,\tilde{m}_{S}, k \cdot \log M)$
		\State perform \textsc{Beep-Wave}$(v,\tilde{m})$
	\end{algorithmic}
\end{algorithm}

\begin{theorem}
	\label{the:MB1}
	\textsc{Multi-Broadcast With Provenance}$(S,m(S))$ correctly performs multi-broadcast with provenance within $O(k \log \frac{LM}{k} + D \log L)$ time-steps.
\end{theorem}

\begin{proof}
	The three sub-procedure calls in the initial `set-up' phase take a total of $O(D \log L + \log M)$ time-steps, and provide a leader node and knowledge of $D$ and $\log M$.
	
	Round $i$ of the main loop of the algorithm takes $O(D+k_i)$ time, since it consists of performing \textsc{CollectMessages} on strings of length $O(k_i)$, and then \textsc{Beep-Wave} on a string of the same length. Furthermore, since the number of known prefixes at most doubles each round, $k_i \le 2^{i-1}$. Hence, there exists some constant $c$ such that total time for the loop is bounded by:
	\begin{align*}
	\sum_{i=1}^{\log L} c (D + k_i)
	& =
	c D \log L + c \left(\sum_{i=1}^{\log k} k_i + \sum_{i=\log k+1}^{\log L} k_i\right)\\
	&\le
	c D \log L + c \left(\sum_{i=1}^{\log k} 2^{i-1} + \sum_{i=\log k+1}^{\log L} k\right)
	\\
	& \le
	c D \log L + c(k + k(\log L - \log k))
	=
	O(D \log L + k \log \frac{L}{k})
	\enspace.
	\end{align*}
	
	The final call to \textsc{CollectMessages} then takes a further $O(D+ k \log M)$ time, and so total running time is $O(D \log L + k \log \frac{L}{k} + k \log M) = O(k \log \frac{LM}{k} + D \log L)$
	
	Correctness follows since each round of the loop informs the leader of the next bit in each ID prefix, and it then broadcasts this information to the network. After $\log L$ rounds, all nodes know all source IDs and each source $s$ can correctly construct its string $\tilde{m}_s$. The \textbf{OR}-superimposition of these strings, broadcast to all nodes, is a list of messages in source ID order, which fulfills the goal of the algorithm.
\end{proof}

\subsubsection{Multi-Broadcast Without Provenance}

It may be the case that we do not need to know where messages originated from, or the number of duplicate messages; for example when using short control messages instructing all nodes to perform some action, for which provenance might be irrelevant. For this reason, we also study the variant of multi-broadcast where nodes need only know one copy of each unique source message, and no source IDs.

The main difference in concept for our multi-broadcast without provenance algorithm (Algorithm \ref{alg:MB2}) is that the concurrent binary searches are performed on the bits of the source messages rather than the IDs. However, this requires $O(D\log M)$ time, which is too slow when $k<D$ and $L<M$, and so we first run Algorithm \ref{alg:MB1}, curtailing it when our number $k_i$ of known ID prefixes (which is a lower bound for $k$) exceeds $D$, in order to efficiently deal with these cases.

If $k \le D$ then the call to algorithm \ref{alg:MB1} will complete multi-broadcast (meeting the requirements for the case without provenance, since they are strictly weaker than those with provenance). Otherwise, we move onto performing binary searches on the bits of the message. This functions in much the same way as in Algorithm \ref{alg:MB1}, except that we do not need the final \textsc{CollectMessages} and \textsc{Beep-Wave} stage since the network is already aware of all source messages upon completion of the main loop. We will use $\tilde{k}_i$ to denote the number of $i-1$-bit message prefixes known to nodes at the start of round $i$ of the for loop, and $\tilde{Z}_{s,i}$ to be the string constructed by source $s$ in round $i$ by placing a \textbf{1} in the position corresponding to the $i$-bit prefix of its message and \textbf{0} in all others.

\begin{algorithm}[h]
	\caption{\textsc{Multi-Broadcast Without Provenance}$(S,m(S))$}
	\label{alg:MB2}
	\begin{algorithmic}
		\State perform \textsc{Multi-Broadcast With Provenance}$(S,m(S))$ until $k_i > D$
		\If {it did not complete}
		\For {$i = 1$ to $\log M$}
		\State $\tilde{Z}_i \gets \textsc{CollectMessages}(v,S,\tilde{Z}_{S,i}, 2\tilde{k}_i)$
		\State perform \textsc{Beep-Wave}$(v,\tilde{Z}_i)$
		\EndFor
		\EndIf
	\end{algorithmic}
\end{algorithm}

\begin{theorem}
	\label{the:MB2}
	\textsc{Multi-Broadcast Without Provenance}$(S,m(S))$ correctly performs multi-broadcast without provenance within $O(k \log \frac{M}{k} + D \log L)$ time-steps if $k<M$, and $O(M + D \log L)$ time-steps if $k \ge M$.
\end{theorem}

\begin{proof}
	By the same argument as for Theorem \ref{alg:MB1}, each round of the main loop informs all nodes of the next bit in each message prefix. Therefore, after $\log M$ rounds we have performed multi-broadcast without provenance.
	
	We separate the running-time proof into four cases:
	\begin{enumerate}[\bf (1) ]
		\item $k \le D$ and $k< M$;
		\item $k \le D$ and $k\ge M$;
		\item $k>D$ and $k< M$;
		\item $k>D$ and $k \ge M$.
	\end{enumerate}
	
	\begin{description}
		\item[Case 1: $k \le D$ and $k< M$.]
		
		For the $k \le D$ case, the number of unique $i$-bit source ID prefixes $k_i$ will never exceed $D$ (since it is bounded above by $k$), and so the all to \textsc{Multi-Broadcast With Provenance} will successfully perform multi-broadcast (with provenance, and therefore also without) in $O(k \log \frac{LM}{k} + D \log L) = O(k \log {L} + k \log \frac{M}{k} + D \log L) = O(k \log \frac Mk+D \log L)$ time-steps.
		
		\item[Case 2: $k \le D$ and $k\ge M$.]
		
		As above, the call to \textsc{Multi-Broadcast With Provenance} will successfully perform multi-broadcast in $O(k \log \frac{LM}{k} + D \log L) = O(k \log {L} + D \log L) = O(D \log L)$ time-steps.
		
		\item[Case 3: $k>D$ and $k< M$.]
		
		Since $k>D$, the call will not complete multi-broadcast, but its ``set-up'' phase will provide a leader $v$ and knowledge of $D$ and $\log M$, so these steps are not duplicated in our description of Algorithm \ref{alg:MB2}. Each round of the main loop then informs every node of the next bit in each unique message prefix, and so after $\log M$ rounds we are done.
		
		Let $t$ be the round of the loop at which the call to \textsc{Multi-Broadcast With Provenance} terminates. Running time for the call is then bounded above (for some constant $c$) by
		\begin{align*}
		c D \log L + \sum_{i=1}^{t} c (D+ k_i)
		&\le
		c D \log L + \sum_{i=1}^{t} 2cD
		\\
		&\le
		c D \log L + \sum_{i=1}^{\log L} 2cD
		\\
		&=
		3 c D \log L
		=
		O(D \log L)
		\enspace,
		\end{align*}
		where the first inequality is due to the fact that $k_i \le D$ until termination.
		
		Running time for the main loop of Algorithm \ref{alg:MB2} is bounded above (again for some constant $c$) by:
		\begin{align*}
		\sum_{i=1}^{\log M} c (D + \tilde{k}_i)
		& =
		c D \log M + c \left(\sum_{i=1}^{\log k} \tilde{k}_i + \sum_{i=\log k+1}^{\log M} \tilde{k}_i\right)
		\\
		&\le
		c D \log M + c \left(\sum_{i=1}^{\log k} 2^{i-1} + \sum_{i=\log k+1}^{\log M} k\right)
		\\
		& \le
		c D \log M + c(k + k(\log M - \log k))
		\\
		&=
		O(D \log M + k \log \frac{M}{k})
		\enspace.
		\end{align*}
		
		Total time is therefore
		\begin{align*}
		O(D \log L + D \log M + k \log \frac{M}{k})
		&=
		O(D \log L + D \log \frac{M}{k} + D \log k + k \log \frac{M}{k})
		\\
		&=
		O(k \log \frac{M}{k} + D \log L)
		\enspace,
		\end{align*}
		where the last expression holds since $D \log k \le D \log L$ and $D \log \frac Mk \le k \log \frac Mk$.
		
		\item[Case 4: $k>D$ and $k \ge M$.]
		
		The call to \textsc{Multi-Broadcast With Provenance} will fail and take $O(D \log L)$ time as before. Running time for the main loop of Algorithm \ref{alg:MB2} is now bounded by:
		
		\begin{align*}
		\sum_{i=1}^{\log M} c (D + \tilde{k}_i)
		& =
		c D \log M + c \sum_{i=1}^{\log M} \tilde{k}_i
		\le
		c D \log M + c \sum_{i=1}^{\log M} 2^{i-1}\\
		&\le
		cD \log M + cM
		=
		O(D \log M + M)
		\enspace.
		\end{align*}
		
		Since $M\le k\le L$, total running time is $O(M + D \log L)$.
	\end{description}
	
	\paragraph*{Combining cases:}
	
	When $M>k$ total running time is $O(k\log\frac Mk + D \log L)$, and when $M\le k$, total running time is $O(M + D \log L)$.
\end{proof}

It may seem nonintuitive that Algorithm \ref{alg:MB2} achieves multi-broadcast in fewer then the $k \log M$ time-steps required for a single node to directly transmit or hear the messages, since this might seem to be a natural lower bound. The improvement stems from implicit compression of the messages within the algorithm's method.

\subsection{Faster Non-Explicit Multi-Broadcast}
A very recent result by Dufoulon, Burman, and Beauquier \cite{-DBB18} improves the running time for leader election in the beep model to an optimal $O(D+\log L)$. This allows them to slightly improve the running time of our explicit multi-broadcast with provenance algorithm (Algorithm \ref{alg:MB1}) to $O(k\log\frac{LM}{k} + D \min\{k,\log L\})$. However, it does not directly lead to significantly faster algorithms, because leader election was not a bottleneck in our analysis. In this section, we show how to exploit this improved leader election procedure to attain an optimal algorithm for multi-broadcast with provenance, and near-optimal for multi-broadcasting without provenance.

We take a different approach from Algorithms \ref{alg:MB1} and \ref{alg:MB2}, using a new type of \emph{superimposed code} to collect information. A superimposed code is a function which maps each element of its domain to a unique binary codeword, in such a way that information can be inferred from the binary \textbf{OR}-superimposition of a set of codewords. In our case, we define a $(k,X)$-\emph{choice superimposed code} which guarantees that, given the superimposition of any $k$ codewords, there are at most $O(k)$ codewords that could have been included in the superimposition (because all of the others have a \textbf 1 where the superimposition has a \textbf 0). 

We will say one binary string $a$ is \emph{dominated} by another string $b$ (denoted $a\preceq b$) if $a_i\leq b_i \forall i$. Our goal now is to show that any superimposition of $k$ codewords dominates at most $O(k)$ others.

\begin{definition}
	A \textbf{$(k,X)$-choice superimposed code of length $\ell$} is an injective function $C:X\rightarrow \{0,1\}^\ell$ such that for every set $K\subseteq X$ with $|K|:=k$, the size of the set $Y = \{u\in X:C(u)\preceq\bigvee_{v\in K}C(v)\}$ is at most $9k$. 
\end{definition}

This set $Y$ is the set of all codewords dominated by the superimposition. Note that while the definition specifies superimpositions of exactly $k$ codewords, the set $Y$ is also size $O(k)$ for any superimposition of fewer than $k$ codewords, since these can be arbitrarily extended to $k$ codewords without reducing the size of $Y$.

\begin{lemma}
	For any $k,X$ with $k\leq |X|$, there exists a $(k,X)$-choice superimposed code of length $9k\ln \frac {|X|}{k}$.
\end{lemma}

\begin{proof}
	We will prove the existence of such a code by a probabilistic method. The idea is that we prove the existence of some combinatorial object by randomly generating a candidate object, and then proving that it satisfies the required criteria with positive probability. Then, some such object must exist. The downside of this type of argument is that it is existential, i.e. does not tell us how to construct the object, and so algorithms making use of the object are non-explicit.
	
	Let $x=|X|$. We randomly generate a candidate code $C:X\rightarrow \{0,1\}^\ell$, by choosing each bit of each code-word independently to be: 
	\begin{itemize}
		\item $\textbf 1 \text{ with probability }\frac{1}{2k}\text{ and \textbf 0 otherwise, for the first }6k\ln \frac xk\text{ bits.}$\\
		\item \textbf 1$\text{ with probability }\frac{1}{2}\text{ and \textbf 0 otherwise, for the last }3k\ln \frac xk\text{ bits. }$
	\end{itemize}
	
	The last $3k\ln \frac xk\geq 3\ln x$ bits are solely to ensure that no two code-words are the same (i.e. $C$ is injective as required), which is the case since the probability that two particular codewords agree on those bits is at most $(\frac 12) ^ {3\ln x} \leq x^{-2.07}$. Taking a union bound over all $\binom x2\le \frac 12 x^2$ pairs of codewords, the probability that any two are the same is at most $\frac 12 x^2 \cdot x^{-2.07} \leq \frac 12$.
	
	For the rest of our analysis we consider only the first $6k\ln \frac xk$ bits. Fix some subset $K\subseteq X$ of size $k$. Clearly for all $u\in K$, $C(u)\preceq\bigvee_{v\in K}C(v)$. The probability that any particular bit $\bigvee_{v\in K}C(v)_i$ is \textbf 0 is at least 
	\[\prod_{v\in K} \Prob{C(v)_i = \textbf 0} \geq \prod_{v\in K} (1-\frac{1}{2k})\geq \prod_{v\in K} 4^{-\frac{1}{2k}} = 4^{-\frac{k}{2k}}= \frac 12\enspace.\]
	
	We can then show that any codeword not in $K$ is unlikely to be dominated by $K$'s superimposition. For any $u\notin K$, the probability that $C(u)\preceq\bigvee_{v\in K}C(v)$ is at most
	\begin{align*}
	\prod_{i\in [\ell]} \Prob{\lnot \left(C(u)_i = \textbf 1\land \bigvee_{v\in K}C(v)_i = \textbf 0\right)} &\leq \prod_{i\in [\ell]} (1-(\frac {1}{2k}\cdot \frac 12))\\
	&\leq \prod_{i\in [\ell]} e^{\frac {-1}{4k}} = e^{\frac{-\ell}{4k}}
	\end{align*}
	
	So, the probability that $|Y\setminus K|\geq 8k$ (i.e. $|Y|\geq 9k$) is at most:
	
	\[\binom{X}{8k}  (e^{\frac{-\ell}{4k}})^{8k} \leq \left(\frac{ex}{8k}\right)^{8k}e^{-2\ell}\leq e^{8k\ln \frac xk - 2\ell}\leq e^{-4k\ln\frac xk}\]
	
	There are at most 
	$\binom xk \leq e^{2\ln\frac xk}$ 
	possible sets $K$, and by a union bound over all of them, the probability some set $K$ does not satisfy the condition is at  most $e^{-2k\ln\frac xk}$. By another union bound, the probability that the codewords are not unique or the condition is not satisfied is at most $e^{-2k\ln\frac xk}+\frac 12 < 1$. Since there is a non-zero probability that $C$ is a valid $(k,X)$-choice superimposed code, such a code must exist.
\end{proof}

We now describe how choice superimposed codes can be used for multi-broadcast, in Algorithm \ref{alg:MB3}. 

We perform the same `set-up' phase as in Algorithms \ref{alg:MB1} and \ref{alg:MB2}, electing a leader and obtaining knowledge of diameter $D$ and message length $\log M$. Using the leader election algorithm of \cite{-DBB18}, though, this only requires $O(D+\log L + \log M)$ time. 

Next, we repeatedly perform \emph{rounds} in which we attempt to collect the source messages, encoded using choice superimposed codes. The rounds have a parameter $j$ which doubles each time, starting at a value such that $j\log \frac {M}{j} = D$ (since the rounds will have running time $\Theta(D+j\log \frac {M}{j})$, and we wish to start with the two factors equal). For each round, let $C_j$ be a $(j,[M])$-choice superimposed code of length $9j\ln\frac Mj$. We perform $\textsc{CollectMessages}(v,S,C_j(m(S)), 9j\ln \frac {M}{j})$ and broadcast the resulting string using \textsc{Beep-Wave}. By the properties of choice superimposed codes, if we have $j\geq k$, then the size of the set $Y$ of dominated codewords is at most $9j$. 

When this is the case, we proceed to a final call of \textsc{CollectMessages}. Each source node creates a string $\tilde{m}_{S}$ of length $|Y|$, where the $b^{th}$ bit of the string corresponds to the $b^{th}$ codeword in $Y$ (in lexicographical order). It sets the bit corresponding to the codeword it used for its message to \textbf1, and all others to \textbf0. \textsc{CollectMessages}, performed on these strings and re-broadcast, then informs all nodes of the codewords (and hence the source messages) in use.

\begin{algorithm}[h]
	\caption{\textsc{Non-Explicit Multi-Broadcast}$(S,m(S))$}
	\label{alg:MB3}
	\begin{algorithmic}
		\State $v \gets \textsc{ElectLeader}$
		\State $D \gets \textsc{EstimateDiameter}(v)$
		\State $\log M \gets \textsc{GetMessageLength}(v,S,m(S))$
		\State Let $j$ satisfy $j\log \frac Mj = D$
		\Repeat 
		\State $Z_j \gets \textsc{CollectMessages}(v,S,C_j(m(S)), 9j\ln \frac {M}{j})$
		\State perform \textsc{Beep-Wave}$(v,Z_j)$
		\State $i \gets |Y|$
		\State $j\gets 2j$
		\Until {$i\leq 9j$}
		\State $\tilde{m} \gets \textsc{CollectMessages}(v,S,\tilde{m}_{S}, i)$
		\State perform \textsc{Beep-Wave}$(v,\tilde{m})$
	\end{algorithmic}
\end{algorithm}

\begin{theorem}
	\label{the:MB3}
	Algorithm \ref{alg:MB3} correctly performs multi-broadcast without provenance within $O(k \log \frac{M}{k} + D+\log L)$ time-steps if $k<M$, and $O(M + D+\log L)$ time-steps if $k \ge M$.
\end{theorem}
\begin{proof}
	The three sub-procedure calls in initial ``set-up'' phase take a total of $O(D + \log L + \log M)$ time-steps, and provide a leader node and knowledge of $D$ and $\log M$.
	
	A round of the algorithm's loop with parameter $j$ takes $O(D+j\log \frac {M}{j})$ time, since it consists of performing \textsc{CollectMessages} on strings of length $O(j\log \frac {M}{j})$, and then \textsc{Beep-Wave} on a string of the same length. The loop terminates when $j\leq k$ (assuming $k\leq \frac M2$; if $k\geq \frac M2$ it terminates when $j\leq \frac M2$, since $|Y|\leq M$). 
	
	Let $j'$ be the initial value of $j$, i.e. $j'\log \frac {M}{j'} = D$. If $j'\geq \min{2k,\frac M4}$, then 
	
	We analyze running time of the loop and final \textsc{CollectMessages} call, separating into three cases: 
	\begin{enumerate}[\bf (1) ]
		\item $j'\geq \min{k,\frac M2}$;
		\item $j'\leq k<\frac M2$;
		\item $j'\leq \frac M2\leq k$;
	\end{enumerate}
	
	\begin{description}
		\item[Case 1: $j'\geq \min{k,\frac M2}$.]
		
		The loop terminates after the first round, taking $O(j'\log \frac {M}{j'}) = O(D)$ time. The final \textsc{CollectMessages} call takes $O(j') = O(D)$ time.
		
		\item[Case 2: $j'\leq k<\frac M2$.]
		
		Total running time of the loop is at most
		\begin{align*}
		\cj\sum\limits_{q=\log j'}^{\log k}&\left(D+2^q\log \frac {M}{2^q}\right)
		\leq \cj(D\log\frac{k}{j'} +\sum\limits_{q=1}^{\log k}(2^q\log M-q2^q))\\
		&\leq \cj(j'\log \frac {M}{j'}\log\frac{k}{j'} + 2^{\log k + 1}\log M - (\log k -1) 2^{\log k + 1})\\
		&\leq \cj(k\log \frac {M}{k} + \log\frac{2M}{k}2^{\log k + 1})\\
		&\leq 5\cj k \log \frac Mk\enspace,
		\end{align*}
		for some constant $\cj$.
		
		The final \textsc{CollectMessages} call takes $O(k)$ time.
		
		\item[Case 3: $j'\leq \frac M2\leq k$.]
		
		Total running time of the loop is at most
		\begin{align*}
		\cj\sum\limits_{q=\log j'}^{\log \frac M2}&\left(D+2^q\log \frac {M}{2^q}\right)
		\leq \cj(D\log\frac{M}{2j'} +\sum\limits_{q=1}^{\log \frac M2}(2^q\log M-q2^q))\\
		&\leq \cj(j'\log \frac {M}{j'}\log\frac{M}{2j'} + 2^{\log M}\log M - (\log M -2) 2^{\log M})\\
		&\leq \cj(\frac M2 + 2M)\\
		&\leq 3\cj M\enspace,
		\end{align*}
		for some constant $\cj$.
		
		The final \textsc{CollectMessages} call takes $O(M)$ time.
		
	\end{description}
	
	Combining these cases, we can see that Algorithm \ref{alg:MB3} performs multi-broadcast without provenance within $O(k \log \frac{M}{k} + D+\log L)$ time-steps if $k<\frac M2$, and $O(M + D+\log L)$ time-steps if $k \ge \frac M2$.
\end{proof}

We can use the same algorithm to perform multi-broadcast with provenance, simply by having each node $v$ append its ID to its source message $m(v)$. Then, messages are drawn from the set $[L]\times[M]$ of size $LM$. Replacing $M$ by $LM$ in the statement and proof of Theorem \ref{the:MB3} gives the following:

\begin{theorem}\label{the:MB4}
	Algorithm \ref{alg:MB3} correctly performs multi-broadcast with provenance within $O(k \log \frac{LM}{k} + D)$ time-steps. \qedhere
\end{theorem}
\section{Lower Bounds}

In this section we give lower bounds for the main communications tasks we have considered: broadcasting (in undirected and directed networks) and multi-broadcast. All of these lower bounds follow a similar approach: given $n$, $D$, $L$, $M$, and for multi-broadcasting $k$, we first fix a network $\net$ with $n$ nodes and diameter $D$. Then, we specify a distribution of input instances by choosing uniformly at random the identifier assignment $ID$ (from the set of all injective functions $[n]\rightarrow [L]$). Since we prove lower bounds against randomized algorithms, we will also assume that nodes have as input a random string $y$ drawn independently from some distribution $Y$. Finally, we choose the input message(s) $m$ at random. The distribution we will choose from depends upon the task for which we give a lower bound: 

\begin{itemize}
	\item For broadcasting we choose a single message $m$ uniformly at random from $[M]$;
	\item For multi-broadcasting without provenance we uniformly choose a size-$k$ subset of messages from $[M]$, which we will denote by $m\in \binom{[M]}{k}$;
	\item For multi-broadcasting with provenance we instead draw $m$ from the product of $k$ independent (possibly non-unique) messages with a size-$k$ subset of IDs, i.e. $m\in [M]^k\times \binom{[L]}{k}$. Source nodes use these IDs rather than those specified previously.
\end{itemize}

To encode node behavior, we will denote by $P^v_t \in \{\textbf B, \textbf L \}$ the behavior of a node $v$ at a time-step $t$, where $\textbf B$ means that $v$ beeps, and $\textbf L$ that it listens. We further denote $P^v_{\leq t}$ the sequence of $v$'s behavior up to time-step $t$.
We will also need to model what $v$ would hear upon listening, which we denote $Q^v_t\in \{\textbf H, \textbf S \}$, where $\textbf H$ means that $v$ would hear a beep (i.e. has a neighbor $u$ with $P^u_t = \textbf B$), and $\textbf S$ that it would hear silence. Likewise, we denote $Q^v_{\leq t}:= \{Q^v_{t'}\}_{t'\leq t}$.

We then note that $v$'s output after time-step $t$ must depend entirely on $ID(v)$, $y$, $Q^v_{\leq t}$, and if $v$ is a source, $m_v$. Our goal now will be to show that if insufficient time has passed, the probability that a node $v$'s output is correct will be $o(1)$. We do this by arguing that $ID(v)$ and $y$ are independent of $v$'s correct output, and that $Q^v_{\leq t}$ provides insufficient information to reliably recover this output.

\subsection{Undirected Networks}
We will first show lower bound for broadcasting and multi-broadcast in undirected networks. The bound for broadcasting will be derived as a special case of the multi-broadcasting bound, so we begin with multi-broadcast without provenance.

The network $\net$ we will use as a lower bound is the following: we place one node in each layer $1$ to $D$, and all other nodes in layer $0$. An edge will be present between nodes $u$ and $v$ if they are in consecutive layers, i.e. $|layer(u)-layer(v)|=1$. (Note that we assume here that $n-D\geq k$, but since we are concerned with asymptotic results, if this is not the case we can simply use $k'=\frac k2$ and $D'=\frac D2$ instead.)

As described above, we choose uniformly at random $y \in Y$, $ID\in [n]\rightarrow [L]$, and $ m\in \binom{[M]}{k}$ to generate our input distribution.

We show a lemma which states effectively that a node further than $t$ steps from the source nodes cannot receive any information about source messages before time-step $t$:

\begin{lemma}\label{lem:indm}
	For a time-step $t$, and for a node $v$ in layer $i$ with $i>t$, $P^v_{\leq t}$ is independent of $m$.
\end{lemma}

\begin{proof}
	We prove the claim by induction on $t$. Trivially it is true when $t=1$, since any node $v$ in layer $i>0$ is not a source, and $P^v_{0}$ is determined based only on $ID(v)$ and $y$, which are independent of $m$. 
	
	Assuming the claim is true for $t=j$, and proving for $t=j+1$, for a node $v$ in layer $i>j+1$, $P^v_{\leq j+1}$ is dependent entirely on $ID(v)$, $y$, and $Q^v_{\leq j}$. $Q^v_{\leq j}$ is dependent only on the values $P^u_{\leq j}$ for neighbors $u$ of $v$, and since these nodes are in layers at least $i-1>j$, these $P^u_{\leq j}$ values are also independent of $m$ by the inductive assumption.
\end{proof}

An easy corollary gives an $\Omega(D)$ lower bound:

\begin{corollary}\label{cor:indm}
	Any multi-broadcast algorithm running on $\net$ has $o(1)$ success probability, conditioned on it terminating in fewer than $T<D-1$ time-steps. Asymptotic behavior refers to when $M\rightarrow\infty$.
\end{corollary}

\begin{proof}
	Consider a node $v$ in layer $D$. The output of $v$ after time-step $T$ must depend entirely on $ID(v)$, $y$, and $Q^v_{\leq T}$. $Q^v_{\leq T}$ depends only on $P^u_{\leq T}$ for neighbors $u$ of $v$, and since these nodes are in layers at least $D-1>T$, by Lemma \ref{lem:indm} this is independent of $m$. So, since $m$ is chosen uniformly from $\binom{[M]}{k}$ independently of $v$'s output, the probability that the output is correct is at most $\binom{M}{k}^{-1}$. 
\end{proof}

We now show the $\Omega(k \log \frac M)$ term of the lower bound by arguing that if an algorithm terminates faster than this, $Q_{\leq T}$ contains insufficient information to correctly recover $m$:

\begin{lemma}\label{lem:indm}
	Any multi-broadcast algorithm running on $\net$ has $o(1)$ success probability, conditioned on it terminating in $T\leq\frac k2 \log \frac Mk$ time-steps.
\end{lemma}

\begin{proof}
	The output of any non-source node $v$ at time-step $T$ must depend entirely on $ID(v)$, $y$, and $Q_{\leq T}$. $ID(v)$ and $y$ are independent of $m$, and $Q_{\leq T}$ takes one of only $2^T\leq \left(\frac{M}{k}\right)^\frac k2$ values.
	
	For each $m$, let $q_m$ maximize $\Prob{OUTPUT_v = m|Q_{\leq T} = q_m}$. Note that $\sum_{m \in [M]} \Prob{OUTPUT_v = m|Q_{\leq T} = q_m} \leq\frac{M}{k}^\frac k2$, since each possible value of $Q_{\leq T}$ contributes at most $1$ in total. Then,
	
	\begin{align*}
	&\Prob{OUTPUT_v \text{ is correct}} = \binom{M}{k}^{-1} \sum\limits_{\textbf m\in\binom{[M]}{k}} \Prob{OUTPUT_v =\textbf m | m = \textbf m}\\
	&\hspace{1cm}\leq \binom{M}{k}^{-1} \sum\limits_{\textbf{m}\in\binom{[M]}{k}} \Prob{OUTPUT_v =\textbf{m} | m =\textbf{m}, Q_{\leq T} = q_\textbf{m}}\\
	&\hspace{1cm}\leq \binom{M}{k}^{-1}\left(\frac{M}{k}\right)^\frac k2\leq  \left(\frac{M}{k}\right)^{-k}\left(\frac{M}{k}\right)^\frac k2 = \left(\frac{M}{k}\right)^{-\frac k2} \enspace.
	\end{align*}
\end{proof}

Since these results were proven on the same input distribution, we can combine them:

\begin{theorem}
	\label{the:BLB}
	For $k<\frac M2$, any multi-broadcast without provenance algorithm running on $\net$ has $o(1)$ success probability, conditioned on it terminating within $\frac 14(D+k\log \frac Mk)$ time-steps. 
\end{theorem}

\begin{proof}
	From Corollary \ref{cor:indm} and Lemma \ref{lem:indm}, an algorithm has $o(1)$ success probability conditioned on it terminating within $\max\{D-1,\frac k2 \log \frac Mk\}\geq\frac 14(D+k\log \frac Mk)$ time-steps.
\end{proof}

With slight adjustments, we can also obtain a lower bound when $k\geq \frac M2$:

\begin{theorem}
	\label{the:BLB2}
	For $k\geq\frac M2$, any multi-broadcast without provenance algorithm running on $\net$ has $o(1)$ success probability, conditioned on it terminating within $\frac 18(D+M)$ time-steps.
\end{theorem}

\begin{proof}
	We follow the same lines as the proof of Theorem \ref{the:BLB}, but when specifying our input distribution we only randomly select messages for $k' = \frac M2$ of the sources (the rest we can choose arbitrarily, and in fact can assume are known by all nodes a priori). Then, we reach the same result as Theorem \ref{the:BLB} for for algorithms terminating within $\frac 14(D+k'\log \frac {M}{k'}) = \frac D4+\frac M8$ time-steps.
\end{proof}

We can also easily adapt for multi-broadcast with provenance:

\begin{theorem}
	\label{the:BLB3}
	Any multi-broadcast with provenance algorithm running on $\net$ has $o(1)$ success probability, conditioned on it terminating within $\frac 14(D+k\log \frac {LM}{k})$ time-steps. 
\end{theorem}

\begin{proof}
	We again follow the same lines as the proof of Theorem \ref{the:BLB}, but when specifying our input distribution we now take source inputs to the product of non-unique messages and unique IDs, i.e. $m$ is drawn uniformly at random from the set $[M]^k \times \binom{[L]}{k}$, which has size $M^k\cdot \binom Lk \geq \left(\frac{LM}{k}\right)^k$. We can then show analogously that conditioning on termination within $\frac k2\log \frac {LM}{k}$ time-steps, probability of correct output is at most $\left(\frac {LM}{k}\right)^{-\frac k2}$. The proof that $D-1$ time-steps are required is unchanged. So, an algorithm has $o(1)$ success probability conditioned on it terminating within $\max\{D-1,\frac k2\log \frac {LM}{k}\}\geq\frac 14(D+k\log \frac {LM}{k})$ time-steps.
\end{proof}

An asymptotically optimal lower bound for broadcasting is a special case of Theorem \ref{the:BLB}:

\begin{theorem}
	\label{the:BLB4}
	Any broadcasting algorithm running on $\net$ has $o(1)$ success probability, conditioned on it terminating within $\frac 14(D+\log M)$ time-steps. 
\end{theorem}

\begin{proof}
	Follows from Theorem \ref{the:BLB}, by setting $k=1$.
\end{proof}

\subsection{Directed Networks}

Next, we show that our algorithm for broadcasting in the directed beep model is also optimal.

Our network $\net$ will be as follows, given parameters $n$ and $D$, we again divide the nodes into $D+1$ \emph{layers}; this time layer $0$ contains only the source node $s$, layers $1$ to $D-1$ each contain a single non-source node, and layer $D$ contains all other nodes. Then we let a \emph{directed} edge $(u,v)$ be present in the network if $layer(u) \ge layer(v) - 1$, with the exception that we do not put edges between pairs of nodes in layer $D$. That is, a node has edges to the node in the next layer, and to all nodes in previous layers.

Consider a fixed broadcasting algorithm running on $\net$ for $t$ time-steps. We will denote variable sets $X^i_t$ to be the set of time-steps at most $t$ in which a node in layer $i$ beeps and no nodes in later layers do. We now show, in effect, that all information a node receives about the source message is contained within these sets:

\begin{lemma}\label{lem:indy}
	For time-step $t$, a layer $i$, and for any node $v$ in layer $j > i$, $P^v_{\leq t}$ is dependent entirely on $ID$, $y$, and $X^{i}_{t-1}$.
\end{lemma}

\begin{proof}
	We prove the claim by induction on $t$. The base case $t=0$ is trivially true, since all non-source nodes' input in the time-step $0$ is included in $ID$ and $y$, and so their choice to beep is also fully dependent on these.
	
	For the inductive step, assuming the claim is true for $t<t'$, we prove for $t'$. $P^v_{\leq t}$ is dependent entirely on $ID$, $y$, and $Q^v_{\leq t'-1}$. This latter term is dependent on the value of $P^u_{\leq t'-1}$ for all in-neighbors $u$ of $v$.
	
	If $i<j-1$,  these values are dependent entirely on $ID$, $y$, and $X^{i}_{t'-2}$ by the inductive assumption. All information in $X^{i}_{t'-2}$ is contained in $X^{i}_{t'-1}$, so the claim holds.
	
	If $i=j-1$, however, the inductive assumption cannot be applied to $P^w_{\leq t'-1}$, where $w$ is the in-neighbor of $v$ in layer $j-1$. In this case, $P^w_{\leq t'-1}$ can be determined entirely from $X^{i}_{t'-1}$ and the values $P^u_{\leq t'-1}$ for all nodes $u$ in layers at least $j$. By the inductive assumption, these values $P^u_{\leq t'-1}$, are dependent entirely upon $ID$, $y$, and $X^{i}_{t'-1}$, hence so are $P^w_{\leq t'-1}$ and $P^v_{\leq t'}$.
\end{proof}

We can node prove our lower bound by arguing that if running time is too short, these sets $X^i_t$ contain insufficient information to recover $m$:

\begin{theorem}
	Any algorithm for broadcasting in directed networks which runs in $o(\frac{D\log M}{\log D})$ expected time has $o(1)$ success probability.
\end{theorem}

\begin{proof}
	We consider a node $v$ in layer $D$. Let $\cj\geq 3$ be an arbitrarily large constant. By Lemma \ref{lem:indy}, the output of such a node $v$ after $T\leq\frac{D\log M}{\cj^2\log D}$ time-steps can be expressed as a function of $ID$, $y$ (both of which independent of $m$), and $X^{i}_{T-1}$, for any $i<D$. We will denote random variable $x^i=|X^{i}_{T-1}|$.
	
	For each $m\in [M]$, let $Xmax^{i}_\textbf m$ be the value of $X^{i}_{T-1}$ which maximizes \[\prob{OUTPUT_v=\textbf m|X^{i}_{T-1}=Xmax^{i}_\textbf m, m=\textbf m}\] subject to $x^i\leq \frac{\log M}{\cj\log D}$. 
	
	There are at most
	\[\sum_{x^i =1}^{\frac{\log M}{\cj\log D}}\binom{T}{x^i}
	\leq \frac{\log M}{\cj\log D} \left(\frac{Te}{\frac{\log M}{\cj\log D}}\right)^\frac{\log M}{\cj\log D} =
	\frac{\log M}{\cj\log D} \left(\frac{eD}{\cj}\right)^\frac{\log M}{\cj\log D}
	\leq 2^\frac{\log M}{\cj}
	= M^\frac 1\cj
	\]
	
	possible values of $X^{i}_{t-1}$ with $x^i\leq \frac{\log M}{\cj\log D}$. Therefore, \[\sum_{\textbf m\in [M]}\prob{OUTPUT_v=\textbf m|X^{i}_{T-1}=Xmax^{i}_\textbf m, m=\textbf m, x^i\leq \frac{\log M}{\cj\log D}} \leq M^\frac 1\cj\enspace,\]since each possible value of $X^{i}_{T-1}$ contributes at most $1$ in total to the sum.
	
	Then, for any $i$, 
	\begin{align*}
	&\prob{OUTPUT_v\text{ is correct}|x^i\leq \frac{\log M}{\cj\log D}}\\ &\hspace{5mm}\leq \frac 1M \sum_{\textbf m\in [M]}\prob{OUTPUT_v=\textbf m|m=\textbf m, x^i\leq \frac{\log M}{\cj\log D}}\\
	&\hspace{5mm}\leq \frac 1M \sum_{\textbf m\in [M]}\prob{OUTPUT_v=\textbf m|X^{i}_{T-1}=Xmax^{i}_\textbf m,m=\textbf m, x^i\leq \frac{\log M}{\cj\log D}}\\
	&\hspace{5mm}\leq \frac{M^\frac 1\cj}{M} = M^{\frac 1\cj-1}\enspace.
	\end{align*}
	
	Now assume for the sake of contradiction that a broadcasting algorithm finishes within expected time $\Exp{T}\leq\frac{D\log M}{\cj^2\log D}$ and succeeds with probability $p\geq \frac 2\cj$. Then, for all $i<D$ we must have that $\Prob{x^i> \frac{\log M}{\cj\log D}}\geq p-M^{\frac 1\cj-1}> \frac p2$, and so $\Exp{x^i}> \frac{p\log M}{2\cj\log D}$. Since the sets $X^i_t$ are disjoint, $\sum_{i<D}x^i\leq T$, and therefore \[\Exp{T} > \frac{pD\log M}{2\cj\log D}=\frac{D\log M}{\cj^2\log D}\enspace.\]
	
	This is a contradiction, and so it must be the case that $\Exp{T} > \frac{D\log M}{\cj^2\log D}$. Since $c$ is arbitrarily large, we have shown that any algorithm with $\Omega(1)$ success probability must take $\Omega(\frac{D\log M}{\log D})$ expected time, and conversely, any algorithm taking  $o(\frac{D\log M}{\log D})$ expected time has $o(1)$ success probability.
\end{proof}

\section{Discussion and Open Problems}
Models for networks of very weak devices, such as the beep model, are growing in popularity as such devices become cheaper and more commercially viable; examples of their use include sensor networks and RFID tagging. Our aim here is to provide the first systematic study of algorithms for global tasks in such a model. Our running times are mostly optimal, with the only major grounds for improvement being an optimal \emph{explicit} multi-broadcast algorithm. However, there are several other interesting aspects of the beep model which could merit further research.

One crucial concern in networks of this type is that \emph{energy} is often highly constrained: we may wish to minimize the amount of times nodes transmit (and possibly even listen; we could introduce a third option of `do nothing' in a time-step). A study of how little energy is required to complete communication tasks in the beep model would be interesting.

Another research direction is further weaken the assumptions of the model, in order to make it as widely applicable as possible. The major assumption remaining is that time-steps are \emph{synchronous}, i.e. that nodes local clocks all `tick' at the same rate and beeps are heard immediately. There are several possible ways of modeling \emph{asynchronicity}, and exploring what can be done in asynchronous beeping networks. One work in this direction is \cite{-HP16}.

\newcommand{\Proc}{Proceedings of the\xspace}
\newcommand{\SODA}{Annual ACM-SIAM Symposium on Discrete Algorithms (SODA)}
\newcommand{\PODC}{Annual ACM Symposium on Principles of Distributed Computing (PODC)}
\newcommand{\DISC}{International Symposium on Distributed Computing (DISC)}
\newcommand{\JALGORITHMS}{Journal of Algorithms}

\bibliography{beep}{}

\end{document}